\documentclass[sigconf]{aamas} 

\usepackage{balance} 

\setcopyright{ifaamas}
\acmConference[AAMAS '23]{Proc.\@ of the 22nd International Conference
on Autonomous Agents and Multiagent Systems (AAMAS 2023)}{May 29 -- June 2, 2023}
{London, United Kingdom}{A.~Ricci, W.~Yeoh, N.~Agmon, B.~An (eds.)}
\copyrightyear{2023}
\acmYear{2023}
\acmDOI{}
\acmPrice{}
\acmISBN{}

\acmSubmissionID{559}

\usepackage{balance}

\usepackage{soul}
\usepackage{url}
\usepackage{amsmath}
\usepackage{booktabs}
\usepackage{amsthm}
\usepackage{subcaption}
\usepackage{multicol}
\usepackage{amsfonts}

\usepackage{algorithm}
\usepackage{algpseudocode}
\usepackage{graphicx}
\usepackage{textcomp}
\usepackage{listings}
\usepackage{multirow}
\usepackage{xcolor}
\usepackage{mathtools}
\usepackage{array}
\usepackage{cancel}
\usepackage{comment}
\usepackage{diagbox}
\usepackage{bbm}

\let\oldReturn\Return
\renewcommand{\Return}{\State\oldReturn}
\DeclareMathOperator{\proj}{proj}
\newcommand{\vct}{\mathbf}
\newcommand{\vctproj}[2][]{\proj_{\vct{#1}}\vct{#2}}

\usepackage{bm}
\usepackage{tikz}
\usetikzlibrary{arrows,automata}
\usetikzlibrary{shapes,shapes.geometric}
\usepackage{float}
\usepackage{acronym}
\usepackage{paralist}
\usepackage{appendix}
\newtheorem{theorem}{Theorem} 
\newtheorem{lemma}[theorem]{Lemma}

\newcommand{\KLD}[2]{\mathbf{D}_{\mathrm{KL}} \left( #1 \| #2 \right) }

\definecolor{darkgreen}{rgb}{0,0.5,0}

\allowdisplaybreaks[3] 

\newif\ifuseboldmathops
\newif\ifuseittextabbrevs
\useboldmathopstrue   

\ifuseittextabbrevs

\newcommand{\ie}{{\it i.e.}}

\else

\newcommand{\ie}{i.e.}

\fi

\ifuseboldmathops
\newcommand{\reals}{\mathbf{R}}

\else
\newcommand{\reals}{\mathbb{R}}

\fi

\ifuseboldmathops

\else

\fi

\ifuseboldmathops
\newcommand{\Expect}{\mathop{\bf E{}}\nolimits}

\else
\newcommand{\Expect}{\mathop{\mathbb{E}{}}\nolimits}

\fi

\ifuseboldmathops

\else

\fi

\newcommand{\argmax}{\mathop{\mathrm{argmax}}}

\newcommand{\sink}{\mathsf{sink}}

\newcommand{\abs}[1]{\lvert#1\rvert}

\newcommand{\dist}[1]{\mathsf{Dist}(#1)}

\newcommand{\supp}{\mbox{Supp}}

\renewcommand{\Pr}{\mathbf{Pr}}

\DeclareMathOperator*{\optmins}{\textrm{min.}}
\DeclareMathOperator*{\optmaxs}{\textrm{max.}}

\DeclareMathOperator*{\optsts}{\textrm{s.t.}}

\theoremstyle{definition}
\newtheorem{definition}{Definition}

\newtheorem{problem}{Problem}
\newtheorem{remark}{Remark}
\newtheorem{assumption}{Assumption}

\acrodef{mdp}[MDP]{Markov decision process}
\acrodefplural{mdp}[MDPs]{Markov decision processes}
\acrodef{scltl}[scLTL]{syntactically co-safe LTL}
\acrodef{dfa}[DFA]{deterministic finite-state automaton}
\acrodef{ssp}[SSP]{stochastic Shortest Path}
\acrodef{ids}[IDS]{intrusion detection system}
\acrodef{os}[OS]{operation system}
\acrodef{cvss}[CVSS]{common vulnerability scoring system}
\acrodef{milp}[MILP]{mixed-integer linear program}
\acrodef{mip}[MIP]{mixed integer programming}
\acrodef{kl}[KL]{Kullback--Leibler}
\acrodef{maxent}[MAXENT]{Maximum Entropy}
\acrodef{irl}[IRL]{Inverse Reinforcement Learning} 

\title{Optimal Decoy Resource Allocation for Proactive Defense in Probabilistic Attack Graphs}

\author{Haoxiang Ma}
\affiliation{
  \institution{University of Florida}
  \city{Gainesville}
  \country{United State}}
\email{hma2@ufl.edu}

\author{Shuo Han}
\affiliation{
  \institution{University of Illinois Chicago}
  \city{Chicago}
  \country{United State}}
\email{hanshuo@uic.edu}

\author{Nandi Leslie}
\affiliation{
  \institution{Raytheon Technologies}
  \city{Arlington County}
  \country{United State}}
\email{nandi.o.leslie@raytheon.com}

\author{Charles Kamhoua}
\affiliation{
  \institution{U.S. Army Research Laboratory}
  \city{Gainesville}
  \country{United State}}
\email{charles.a.kamhoua.civ@mail.mil}

\author{Jie Fu}
\affiliation{
  \institution{University of Florida}
  \city{Gainesville}
  \country{United State}}
\email{fujie@ufl.edu}

\begin{abstract}
This paper investigates the problem of synthesizing proactive defense systems in which the defender can allocate deceptive targets and modify the cost of actions for the attacker who aims to compromise security assets in this system. We model the interaction of the attacker and the system using a formal security model-- a probabilistic attack graph. By allocating fake targets/decoys, the defender aims to distract the attacker from compromising true targets. By increasing the cost of some attack  actions, the defender aims to discourage the attacker from committing to certain policies and thereby improve the defense. To optimize the defense given limited decoy resources and operational constraints, we formulate the synthesis problem  as a bi-level optimization problem, while the defender designs the system, in anticipation of the attacker's best response given that the attacker has disinformation about the system due to the use of deception. Though the general formulation with bi-level optimization is NP-hard, we show that under certain assumptions, the problem can be transformed into a constrained optimization problem. We proposed an algorithm to approximately solve this constrained optimization problem using a novel, incentive-design method for projected gradient ascent. We demonstrate the effectiveness of the proposed method using extensive numerical experiments.
\end{abstract}


\newcommand{\BibTeX}{\rm B\kern-.05em{\sc i\kern-.025em b}\kern-.08em\TeX}

\begin{document}

\pagestyle{fancy}
\fancyhead{}

\maketitle

\section{Introduction}

Proactive defense refers to a class of defense mechanisms for the defender to detect any ongoing attacks, distract the attacker with deception, or use randomization to increase the difficulty of an attack to the system. In this paper, we propose a mathematical framework and solution approach for synthesizing a proactive defense system with  deception.

We start by formulating the attack planning problem using a probabilistic attack graph, which can be viewed as a \ac{mdp} with a set of attack target states. Attack graphs(AGs)\cite{jhaTwoFormalAnalyses2002} can be used in modeling computer networks. They are widely used in network security to identify the minimal subset of vulnerability/sensors to be used in order to prevent all known attacks\cite{noel2008optimal,sheyner2002automated}. Probabilistic attack graphs introduce uncertain outcomes of attack actions that account for action failures in a stochastic environment. For example, in~\cite{hongAssessingEffectivenessMoving2016,hongHARMsHierarchicalAttack2012},   probabilistic transitions in attack graphs  capture uncertainties originated from network-based randomization. Under the probabilistic attack graph modeling framework, we investigate how to allocate decoy resources as fake targets to distract the attacker   into attacking the fake targets, and how to modify the attack action costs to discourage the attacker from reaching the true targets.

The joint design of decoy resource allocation and action cost modification can be cast as a bi-level optimization problem, which is generally NP-hard~\cite{dempe2020bilevel}. Under the assumption that potential decoy states are predefined and the defender only needs to allocate resources/rewards to decoys, we prove the bi-level optimization can be equivalently expressed as a constrained optimization problem. To solve the constrained optimization problem using a projected gradient ascent efficiently, we build two important relations: First, we show that the projection step of a defender's desired attack policy to the set of realizable attack policy space can be performed using  \ac{irl} \cite{ziebart2008maximum}. Essentially,   \ac{irl} is to shape the attacker's \emph{perceived reward} so that the rational attacker will mimic a  strategy chosen by the defender. Second, the gradient ascent step can be performed using  policy improvement, which is a subroutine in policy iteration with respect to maximizing the defender's total reward. The project gradient ascent is ensured to converge to a (local) optimal solution to this  nonconvex constrained optimization problem.


\paragraph*{Related work}
The synthesis of proactive defense strategies studied here is closely related to the Stackelberg security game(SSG) (surveyed in \cite{sinhaStackelbergSecurityGames2018}) and its solution via bi-level optimization. In an SSG, the defender is to protect a set of targets with limited resources, while the attacker selects the optimal attack strategy given the knowledge of the defender's strategy.   In \cite{nguyenMultistageAttackGraph2018}, the authors study security countermeasure-allocation and   use attack graphs  to evaluate the network's security given the allocated resources. 
However, the SSG does not account for the asymmetric information introduced by the use of deception.   In \cite{wu2019reward}, the authors introduce reward shaping to motivate the agent to behave as the target policy. However, in our setting, the target policy may be infeasible, because the defender aims to lure the attacker to reach a fake target, while the attacker may not intentionally avoid true targets.


Deceptions create incorrect/incomplete information to the attacker. In \cite{thakoorCyberCamouflageGames2019}, the authors formulate a security game to allocate limited decoy resources to mask a network configuration from the cyber attacker. The decoy-based deception manipulates the adversary's perception of the payoff matrix. In \cite{Anwar2020}, the authors study honeypot allocation   in  deterministic attack graphs and determine the optimal allocation strategy using the minimax theorem.  In \cite{milaniHarnessingPowerDeception2020}, the authors study  directed acyclic attack graphs that can be modified by the defender using deceptive and protective resources. They propose a \ac{milp}-based algorithm to  determine the allocation of deceptive and protective resources in the graph. In \cite{durkota2015optimal}, they harden the network by using honeypots so that the attacker can not discriminate between a true target and a fake target.  In \cite{milani2020harnessing}, the authors assign fake edges in the attack graph in order to interdict the attacker and employ MILP to find the optimal solution. 

Compared to existing work, our work makes the following contributions: First, we do not assume any graph structure in the attack graph and consider probabilistic attack graphs instead of deterministic ones. As the attacker can take a randomized strategy in the probabilistic attack graph, it is not possible to construct a payoff matrix and apply the minimax theorem for decoy resource allocation.  Second, we consider simultaneously allocating limited decoy resources and modifying the cost of attack actions and analyze the best response of the attacker given the disinformation caused by deception.
Third,   we proposed an efficient incentive-design inspired algorithm for synthesizing the defense strategy Under the assumption that the attacker is rational and can not distinguish decoys from the true targets, by modifying the action reward and allocating decoy resources properly, we show that it is possible to shape the attacker's behavior so that the misperceived attacker is incentivized to commit an attack strategy that maximizes the defender's reward.   Finally, we test the scalability of our method on different problem sizes.

\section{Preliminaries and Problem Formulation}
\paragraph*{Notations}
Let $\reals$ denote the set of real numbers and $\reals^n$ 
the set of real $n$-vectors. 
Let $\reals^n_{>0}$ (resp. $\reals^n_{<0}$) be the set of positive (resp. negative) real $n$-vectors.
We use $\mathbf{1}$ to represent the vector of all ones. Given a vector $z \in \reals^{n}$, let $z_i$ be the $i$-th component. Given a finite set $Z$, the set of probability distributions over $Z$ is represented as $\dist{Z}$. Given $d \in \dist{Z}$, the support of $d$ is denoted as $\supp(d)= \{z\in Z\mid d(z)>0\}$. Let $I_B$ be the indicator function, \ie, $I_B(x) = 1$ if $x \in B$, and $I_B(x)=0$ otherwise.

We consider the adversarial interaction between a defender (player 1, pronoun she/her) and an attacker (player 2, pronoun he/him/his) in a  system equipped with proactive defense (formally defined later). 
We first introduce a formal model, called probabilistic attack graph, to capture  how the attacker plans  to achieve the attack objective. Then, we  introduce proactive defense countermeasures with  deception.

\paragraph*{Attack Planning Problem} The attack planning problem   is modeled as a probabilistic attack graph,
\[
M = (S, A, P, \nu, \gamma, F, R_2),
\]
where $S$ is a set of states (nodes in the attack graph), $A$ is a set of attack actions,  $P: S \times A \rightarrow \dist{S}$ is a probabilistic transition function such that $P(s'|s, a)$ is the probability of reaching state $s'$ given action $a$ being taken at state $s$, $\nu \in \dist{S}$ is the initial state distribution, $\gamma \in (0,1]$ is a discount factor. The attack's objective is described by a set $F$ of \emph{target states} and a \emph{target reward} function $R_2: F \times A \rightarrow \reals_{\ge 0}$, which assigns each state-action pair $(s, a)$ where $s \in F$ and $a\in A$ to a nonnegative value of reaching that target for the attacker. 
The reward function can be extended to the entire state space by defining $R_2(s, a)=0$ for any $s\in S\setminus F, a \in A$. 
To capture the termination of attacks, we introduce a unique sink state $s_\sink \in S\setminus F$ such that $P(s_\sink|s_\sink, a) =1$ for all $a\in A$ and $P(s_\sink|s, a) =1$ for any target $s\in F$ and $a\in A$.

The probabilistic attack graph  characterizes goal-directed attacks encountered in cyber security \cite{lallieReviewAttackGraph2020a,noel2010measuring}, in which by reaching a target state, the attacker compromises certain critical network hosts. Probabilistic attack graphs~\cite{singhal2017security,milaniHarnessingPowerDeception2020} capture the uncertain outcomes of the attack actions using the probabilistic transition function and generalize deterministic attack graphs \cite{jhaTwoFormalAnalyses2002}.

The attacker is to maximize his discounted total reward, starting from the initial state $S_0 \sim \nu$. A randomized, finite-memory attack policy is a function $\pi\colon S^\ast \rightarrow \dist{A}$, which maps a finite run $\rho\in S^\ast $ into a distribution $\pi(\rho)$ over actions. 
A policy is called Markovian   if it only depends on the most recent state, \ie, $\pi \colon S\rightarrow \dist{A}$. 
We only consider Markovian policies because it suffices to search within Markovian policies for an optimal attack policy.

Let $(\Omega, \mathcal{F})$ be the canonical sample space for $(S_0, A_0, (S_t, A_t)_{t>1})$ with the Borel  $\sigma$-algebra $\mathcal{F}=\mathcal{B}(\Omega)$ and $\Omega = S\times A\times \prod_{t=1}^\infty(S\times A)$.  The probability measure $\Pr^\pi$ on $(\Omega, \mathcal{F})$ induced by a Markov policy $\pi$ satisfies: $\Pr^\pi(S_0=s) = \mu_0(s)$, $\Pr^\pi(A_0=a\mid S_0=s)  =\pi(s,a)$, and $\Pr^\pi(S_t=s\mid (S_k, A_k)_{k<t}) = P(s\mid S_k ,A_k )$, and $\Pr^\pi(A_t=a\mid (S_k, A_k)_{k<t}, S_t) = \pi(S_t,a)$.

Given a Markovian policy $\pi \colon S \to \dist{A}$, we define the attacker's value function  $V_{2}^{\pi}: S \rightarrow \reals$ as
\[
V_{2}^{\pi}(s) = \Expect_{\pi}[\sum\limits_{k = 0}^{\infty}\gamma^{k}R_2(S_k, A_k)|S_0 =s],
\] where $\Expect_{\pi}$ is the expectation given the probability measure $\Pr^\pi$.

 \paragraph*{Proactive Defense with Deception} 
 We assume that the defender knows the attacker's objective given by the tuple $\langle  F, R_2\rangle$, \ie, the target states and target reward function. The defender's proactive defense mechanisms are the following:
 
 \begin{itemize} 
 \item Defend by deception: The defender employs a deception method called ``revealing the fake''. Specifically, the defender has a set $D   \subset S \setminus F $ of states in the \ac{mdp} $M$  that can be set to be \emph{fake target states} with  fake target rewards $\vec{y}\in \reals^{\abs{D}}$.  The attacker cannot distinguish the real targets $F$ from fake targets $D$. 
    \item Defend by state-action reward modification: The defender has a set $W \subset (S \setminus (F\cup D)) \times A$ of state action pairs in the MDP $M$ whose reward can be modified. Once the reward of the state action pair $(s, a)$ is modified, the attacker's perceived reward $R_2(s, a) < 0$, \ie, the cost of attack action $a$ at state $s$ is $-R_2(s, a)$.
    \end{itemize}
    The defender can determine how to allocate her decoy resource and limited state-action reward modification ability.

\begin{definition}[Decoy allocation   under constraints]
The defender's decoy allocation design is a nonnegative  real-valued vector $\vec{y}  \in \reals_{\ge 0}^{\abs{S}}$ satisfying $\vec{y}(s) = 0$ for any $s\in S\setminus D$ and constrained by   $\mathbf{1}^{\mathsf{T}}\vec{y} \le h$ for some $h\ge 0$. Given a decoy allocation   $\vec{y}$, the attacker's \emph{perceptual reward function} is defined by 
\begin{equation*}
\label{eq:misperceived-reward-decoy}
     R_2^{\vec{y}}(s, a) =  \left\{
    \begin{array}{ll}
    \vec{y}(s) & \text{ if }  \vec{y}(s) > 0,\\
R_2(s, a) & \text{ if } \vec{y}(s) = 0.
\end{array}\right.
\end{equation*}
\end{definition}

\begin{definition}[Action reward modification]
Given a set $W \subset (S \setminus (F\cup D)) \times A$, the defender's action reward modification is a nonpositive reward-valued vector $\vec{x} \in \reals_{\le 0}^{\abs{S \times A}}$ satisfying $\vec{x}(s, a)  = 0$ for any $(s, a) \notin W$. Given an action reward modification $\vec{x}$, the attacker's perceptual reward function is defined by
\begin{equation*}
\label{eq:misperceived-reward-action}
     R_2^{\vec{x}}(s, a) =  \left\{
    \begin{array}{ll}
    \vec{x}(s, a) & \text{ if }  \vec{x}(s, a) < 0,\\
R_2(s, a) & \text{ if } \vec{x}(s, a) = 0.
\end{array}\right.
\end{equation*}
\end{definition} 

Note that the defender does not consider modifying the state-action reward for (fake or real) target states $F\cup D$ because once a state in $F\cup D$ is reached, the attack is terminated.


\begin{definition}
  The defender's \emph{proactive defense strategy} is a tuple $(\vec{x},\vec{y})$ including an action reward modification $\vec{x}$ and a decoy allocation design $\vec{y}$. 
\end{definition}

Because the action reward modification is independent of the decoy allocation design, the reward function given a defender's strategy $(\vec{x},\vec{y})$ is the composition of $R_2^{\vec{x}}$ and $R_2^{\vec{y}}$ and thus omitted.

\begin{assumption} The attack process terminates under two cases: Either the attack succeeds, in which the attacker reaches a target  $s\in F$, or the attack is interdicted, in which the attacker reaches a state allocated with a decoy.
\end{assumption}

Our problem can be informally stated as follows. 

\begin{problem}
\label{pro: Optimization}
In the attack planning scenario we mentioned above, determine the defender's  strategy to allocate decoy resources and modify action reward so as to maximize the probability that the attacker reaches a fake target  given the best response of the attacker.  

\end{problem}

\section{Main Results}
In this section, we first define the attacker's perceptual planning problem for a fixed action reward modification and decoy resource allocation    $(\vec{x}, \vec{y})$. Then we show that the design of the proactive defense can be formulated as a bi-level optimization problem. We investigate the special property of the formulated bi-level optimization problem to develop an optimization-based approach for synthesizing the proactive defense strategy.


\subsection{A Bi-level Optimization Formulation} 
The defender's  strategy changes how the attacker perceives the attack planning problem as follows:

\begin{definition}[Perceptual attack planning problem with modified reward and decoys]
Let the action reward modification be $\vec{x}$ and decoy allocation be $\vec{y}$, and the attacker's original   planning problem  $M= (S, A, P, \nu, \gamma, F, R_2)$, the perceptual planning problem of the attacker is defined by the following \ac{mdp} with terminating states:
\[
 M(\vec{x},\vec{y}) = (S, A,  P^{\vec{y}},   \nu, \gamma,F \cup D^{\vec{y}}, R_2^{\vec{x},\vec{y}}),
\]
where $S,A, \nu, \gamma$ are the same as those in $M$, $D^{\vec{y}}=\{s \in D \mid \vec{y}(s)\ne 0\}$ are decoy target states and absorbing. The transition function $P^{\vec{y}}$ is obtained from the original transition function $P$ by only making all states in $D^{\vec{y}}$ absorbing. The reward $R_2^{\vec{x},\vec{y}}$ is defined based on Def.~\ref{eq:misperceived-reward-decoy} and Def.~\ref{eq:misperceived-reward-action}.
\label{def:AttackPerceptualMDP}
\end{definition}

The perceptual value for the attacker is \[
V_2^{\pi}(\nu;\vec{x},\vec{y}) = \Expect_\pi \Bigl[
\sum_{k=0}^{\infty} \gamma^k R_2^{\vec{x},\vec{y}}(S_k, A_k) \mid S_0 \sim \nu \Bigr],
\]
where $\Expect_\pi$ is the expectation given the probability measure $\Pr^\pi$ in duced by $\pi$ from the \ac{mdp} $M(\vec{x},\vec{y})$.  

The defender's deception objective is given by a  reward function 
 $R_1^{\vec{y}}:S\rightarrow \reals $, defined by 
\begin{equation}
\label{eq:defender_reward}
       R_1^{\vec{y}}(s) =\begin{cases}
    1 & \text{if $s\in D^{\vec{y}}$},\\
    0 & \text{otherwise}.
        \end{cases}
\end{equation}

Given the probability measure $\Pr^\pi$, we denote the defender's value by \[ 
    V_1^{\pi}(\nu;\vec{y} ) = \Expect_\pi \Bigl[
\sum_{k=0}^\infty \gamma^k R_1(S_k)\mid S_0\sim \nu 
\Bigr].
\] 

With this reward definition, the value  $V_1^\pi(\nu; \vec{y})$ is the probability of the attacker reaching a fake target in $D^{\vec{y}}$.

To formalize the deception objective, we introduce the notion of a defender's preferred attack policy as follows.

\begin{definition}[A defender's preferred attack policy]
Given the perceptual   planning problem of the attacker $M(\vec{x},\vec{y})$ where $(\vec{x},\vec{y})$ is a fixed proactive defense strategy, 
let  $\pi$ and $\pi'$ be two attack policies that achieve the same value for the attacker, \ie, $V_2^\pi(\nu; \vec{x},\vec{y}) = V_2^{\pi'}(\nu; \vec{x},\vec{y})$. Policy $\pi$ is strictly preferred to $\pi'$ by the defender if and only if 
\[
V_1^\pi(\nu; \vec{y} ) > V_1^{\pi'}(\nu; \vec{y}).
\]
 \end{definition}
 In words, if two policies are equally good for the attacker, the one with a higher probability to reach a fake target is preferred by the defender.

Then the problem of synthesizing an optimal proactive defense strategy $(\vec{x},\vec{y})$ can be mathematically formulated as 
\begin{problem}
\begin{alignat*}{2}
 & \optmaxs_{\vec{x}\in X,\vec{y}\in Y} &  & V_{1}^{\pi^{\ast}}(\nu ; \vec{y})\\\
 & \optsts &\quad  & \pi^{\ast}\in \argmax_\pi   V_{2}^{\pi}(\nu;\vec{x} ,\vec{y}   ). 
\end{alignat*}
\label{pro:subproblem2}
where $X =\reals^{|W|}_{\le 0}$ and $Y =\{\vec{y}\mid\forall s\in S\setminus D, \vec{y}(s)= 0 \text{ and }\mathbf{1}^{\mathsf{T}}\vec{y} \le h  \}$ are the ranges for variables $\vec{x}$ and $\vec{y}$ correspondingly.
\end{problem}

 In words, the defender decides $(\vec{x},\vec{y})$ so that  the attacker's  best response in his perceptual attack planning problem turns out to be an attack policy most preferred by the defender, as it maximizes the defender's value.

\subsection{Transforming into a Constrained Optimization Problem}

The bi-level optimization problem is known to be strongly NP-hard \cite{hansen1992new}.
However, under certain conditions, the bi-level optimization problem can be shown to be equivalent to a constrained optimization problem. 



Let $\Pi(\vec{x}, \vec{y}) =  \{\pi\mid V_{2}^{\pi}(\nu;\vec{x} ,\vec{y}) = \max_\pi    V_{2}^{\pi}(\nu;\vec{x} ,\vec{y})\}$
, which is the set of optimal policies in the attacker's perceived planning problem with respect to  a choice of variables $\vec{x}$ and $\vec{y}$. The bi-level optimization problem is then equivalently written  as the following constrained optimization problem:
\begin{alignat}{2}
 & \optmaxs_{\pi^\ast, \vec{x} \in X,\vec{y} \in Y} &  & V_{1}^{\pi^{\ast}}(\nu;  \vec{y} ) \nonumber \\
 & \optsts &\quad  & \pi^{\ast}\in  \Pi(\vec{x}, \vec{y}). \label{pro:bilevel-1}
\end{alignat}
This, in turn, is equivalent to \begin{alignat}{2}
 & \optmaxs_{\pi^\ast} &  & V_{1}^{\pi^{\ast}}(\nu;  \vec{y} ) \nonumber \\
 & \optsts &\quad  & \pi^{\ast}\in \bigcup_{\vec{x} \in X, \vec{y} \in Y} \Pi(\vec{x}, \vec{y}). \label{pro:bilevel-2}
\end{alignat}
Here, the constraint means the attacker's response $\pi^\ast$ can be selected from the collection of optimal attack policies given  all possible values for $\vec{x}$, $\vec{y}$.

By the definition of the defender's value function, it is noted that $V_1^\pi(\nu; \vec{y})$ does not depend on the exact value of $\vec{y}$ but only depends on whether $\vec{y}(s)>0$ for each state $s\in D$. Formally,
\begin{lemma}
For any $\vec{y}_1,\vec{y}_2 \in Y$, if $\vec{y}_1 (s)= 0 \implies \vec{y}_2(s)=0$ and vice versa, then $V_1^\pi(\nu; \vec{y}_1) = V_1^\pi(\nu;\vec{y}_2)$.
\end{lemma}
\begin{proof}
Given two different vectors $\vec{y_1}$ and $\vec{y_2}$, we can construct two \ac{mdp}s: $M_1 \coloneqq M (\vec{x},\vec{y_1}) = (S, A, P^{\vec{y_1}}, \nu, \gamma, F, R_1)$ and $M_2 \coloneqq M(\vec{x},\vec{y_2}) = (S,A, P^{\vec{y_2}}, \nu, \gamma, F, R_1)$, respectively. 

If $\vec{y_1}(s) = 0 $ if and only if $\vec{y_2}(s) = 0$, then the transition functions $P^{\vec{y}_1}$ of $M_1$ and $P^{\vec{y}_2}$ of $M_2$ are the same (see Def.~\ref{def:AttackPerceptualMDP}). 

Further, the defender's reward function $R_1^{\vec{y}_1} $ also equals to $R_1^{\vec{y}_2} $ (see \eqref{eq:defender_reward}), given both the transition dynamics and reward are the same, we have  $V_1^\pi(\nu; \vec{y}_1) = V_1^\pi(\nu;\vec{y}_2)$.
\end{proof}

Next, to remove   the dependency of $V_1^\pi(\nu; \vec{y})$ on $\vec{y}$, we make the following assumption:

\begin{assumption}
The set $D^{\vec{y}} = \{s\in D\mid \vec{y}(s)\ne 0\}$ of states where decoys are allocated is given.
\end{assumption}
Under this assumption, we simply assume all states in the given set $D$ have to be assigned with nonzero decoy resources. That is $D^{\vec{y}} =D$.  

This assumption further reduces the defender's synthesis problem into    a constrained optimization problem.

\begin{alignat}{2}
 & \optmaxs_{\pi^\ast} &  & V_{1}^{\pi^{\ast}}(\nu) \nonumber \\
 & \optsts &\quad  & \pi^{\ast}\in  \overline{\Pi} \triangleq   \bigcup_{\vec{y} \in Y, \vec{x}\in X} \Pi(\vec{x}, \vec{y}), \nonumber \\
  & &\quad & \vec{y}(s) >0, \forall s \in D. \label{eq:decoy-pos} 
\end{alignat}

Because the above problem is a standard constrained optimization problem, one can obtain a locally
optimal solution using the projected gradient method:
\[
\pi^{k+1} =  \vctproj[\overline{\Pi}](\pi^{k} + \eta\nabla V_{1}^{\pi^{k}}(\nu)).
\]
where $\vctproj[\overline{\Pi}](\pi)$ denotes projecting policy $\pi$ onto the policy space $\overline{\Pi}$ and $\eta$ is the step size.  

\subsection{Connecting Inverse-reinforcement Learning with Project Gradient Ascent}
A key step in performing projected gradient ascent is to evaluate, for any policy $\hat \pi$, the projection $ \vctproj[\overline{\Pi}](\hat \pi)$. However, this is nontrivial because the set $\bar \Pi$ includes a set of attack policies, each of which corresponds to a choice of vectors $(\vec{x},\vec{y})$. As a result, $\bar \Pi$
does not have a compact  representation. Next, we propose a novel algorithm that computes the projection. 

First, by the definition of projection, it is noted that 
this projection  step is equivalent to solving the following optimization problem:

\begin{alignat}{2}
 & \optmins_{\pi}  &  & \mathbf{D}(\pi,  \hat{\pi} ) \nonumber \\ 
 & \optsts &\quad  & \pi\in \overline{\Pi},\nonumber \\
   & &\quad & \vec{y}(s) >0; \forall s \in D.
\label{eq:gradient_descent_step}
\end{alignat}
where $\mathbf{D}(\pi,  \hat{\pi})$ is the  distance between the two policies $\pi, \hat \pi$.


The distance function $\mathbf{D}$ can be chosen to be the \ac{kl}-divergence between policy-induced Markov chains, defined as follows.
\begin{definition}
Given an \ac{mdp} $M=(S,A,P,\nu)$ and two Markovian policies $\pi_1$, $\pi_2$. Let $M_{\pi_1}=(S, P_1,\nu)$ and $M_{\pi_2}=(S, P_2,\nu)$ be two Markov chains induced from $M$ under $\pi_1$ and $\pi_2$, respectively. The \ac{kl} divergence $
\KLD{M_{\pi_1}}{M_{\pi_2}}  $ (relative entropy from $M_{\pi_2}$ to $M_{\pi_1}$) is defined by
\[
\KLD{M_{\pi_1}}{M_{\pi_2}} = \sum\limits_{\rho \in S^\ast }\Pr_1(\rho) \log\frac{\Pr_1(\rho)}{\Pr_2(\rho)},
\]
where $\Pr_i(\rho)$ is the probability of a path $\rho$ in the Markov chain $M_{\pi_i}$ for $i=1,2$.
\end{definition}

The \ac{kl} divergence in~\eqref{eq:gradient_descent_step} can be expressed as 
\begin{multline}
\label{eq:kl-expression}
\KLD{M_{\pi}(\vec{x} ,\vec{y})}{M_{\widehat \pi }(\vec{x} , \vec{y})}  = \sum_{\rho} \widehat \Pr(\rho ) \log \frac{\widehat \Pr(\rho )}{\Pr(\rho |\vec{x}, \vec{y})}\\
= \sum_{\rho} \widehat \Pr(\rho ) \log\widehat \Pr(\rho ) - \sum_{\rho} \widehat \Pr(\rho ) \log\Pr(\rho |\vec{x}, \vec{y}),
\end{multline}
where $\widehat \Pr(\rho) $ is the probability of path $\rho$ in the Markov chain $M_{\widehat \pi }(\vec{x} , \vec{y})$, and $\Pr(\rho | \vec{y})$  is   the probability of path $\rho$ in the Markov chain $M_{\pi}(\vec{x} ,\vec{y})$ induced by a policy $\pi$.

Because the first term in the sum in~\eqref{eq:kl-expression} is a constant for $\hat \pi$ is fixed, the \ac{kl} divergence minimization problem is equivalent to the following maximization problem: 
\begin{alignat}{2}
\label{eq:decoy-irl}
&\optmaxs_{\vec{x}\in X, \vec{y}\in Y} &\quad & \sum_{\rho} \widehat \Pr(\rho ) \log \Pr(\rho |\vec{x}, \vec{y})\\
&\optsts & &   \vec{y}(s) >0; \forall s \in D  ,\\
  & & & \mathbf{1}^{\mathsf{T}}\vec{y} \le h. \label{eq:decoy-resource}
\end{alignat}
Problem~\eqref{eq:decoy-irl} can be solved by an extension of the \ac{maxent}  \ac{irl} algorithm~\cite{ziebart2008maximum}, which was originally developed in the absence of constraints. It is well-known that \ac{irl} is to infer, from the expert demonstration, a reward function for which the expert policy generating the demonstrations is optimal. The use of \ac{irl} to perform the projection is intuitively understood as follows: The goal is to compute a pair of vectors $(\vec{x},\vec{y})$ that alters the attacker's perceived reward function so that the attacker's optimal policy given $(\vec{x},\vec{y})$ is closed to the ``expert policy'' $\hat \pi$, under the constraints of $\vec{y}$.

To handle the decoy resource constraint~\eqref{eq:decoy-resource}, we approximate the constraint using a logarithmic barrier function and
compute the optimal solution $\vec{y}^\ast$ using gradient-based numerical optimization.

Considering the constraint $\mathbf{1}^{\mathsf{T}}\vec{y} \le h$, we implement the barrier function in order to approximate the inequality constraints and rewrite the optimization problem as:
 \begin{align*}
     &\max\limits_{\vec{x},\vec{y}}\sum\limits_{\rho}\widehat \Pr(\rho ) \log \Pr(\rho |\vec{x}, \vec{y}) + \frac{1}{t}\log(h - \mathbf{1}^{\mathsf{T}}\vec{y}) \\
     &\text{ subject to: } \vec{y}(s)=0, \quad \forall s\in S\setminus D.
 \end{align*}
where $t$ is the weighting parameter of the logarithmic barrier function. In our experiment, $t$ is fixed to be $1000$. 
 
 Since constraint $\vec{y}(s)=0, \forall s\in S\setminus D$, can be   incorporated into the domain of decision variables $\vec{y}$, we can use gradient ascent to obtain the optimal $\vec{x}^{\ast}, \vec{y}^{\ast}$ that maximizes the objective function.
Specifically, $\vec{x}$ and $\vec{y}$ can be updated via $\vec{x}^{k+1} = \vctproj[X](\vec{x}^{k} + \eta_{x}\nabla L(\vec{x}, \vec{y}))$, $\vec{y}^{k+1} =\vctproj[Y](\vec{y}^{k} + \eta_{y}\nabla L(\vec{x}, \vec{y}))$.

\subsection{Policy Improvement for Gradient Ascent Step}

After the projection step to obtain a policy $\pi^k$ and the corresponding vector $(\vec{x},\vec{y})$, we aim to compute a one-step gradient ascent to improve the objective function's value 
\[V_1^{k+1} (\nu)= V_{1}^{k}(\nu) + \nabla V_{1}^{k}(\nu), \]
where $V_1^k(\nu)$ is the defender's value evaluated given the  attack policy $\pi^k$ at the $k$-th iteration.

For this step, we perform a policy improvement step with respect to the defender's reward function $R_1^{\vec{y}}$, which now is independent of $\vec{y}$ because the set $D^{\vec{y}}$ is fixed to be a constant set $D$. It is shown in \cite{puterman2014markov,Omid2022aaai} that policy improvement is a one-step Newton update of optimizing the value function.  

Specifically, the policy improvement is to compute 
\[
\tilde \pi^{k+1}(s,a) = \frac{\exp{((R_1(s, a)+\gamma V_1^{k}(s'))/\tau)}}{\sum_{a \in A}\exp{((R_1(s, a)+\gamma V_1^{k}(s'))/\tau)}},
\]


The policy at iteration $k+1$ is obtained by performing the projection step (\eqref{eq:gradient_descent_step}) in which $\hat \pi \triangleq \tilde \pi_{k+1}$.

The iteration stops when $|V_1^{k+1}(\nu) - V_1^{k}(\nu)| \le \epsilon$ where $\epsilon$ is a manually defined threshold.  
The output yields a tuple $(\vec{x}^\ast, \vec{y}^\ast)$ which is the (local) optimal proactive defense strategy. We can only obtain a local optimal proactive defense strategy here due to the transferred constrained optimization problem having a nonconvex constraint set. However, we can start from different initial policies and select the best one. Moreover, assume the defender is solving her own problem without considering attacker's objective. the upper bound of the defender's objective can be obtained. We can select the solution whose objective function is closest to the upper bound. 

\begin{remark}
In our problem, we assume the set $D$ is given. If the set $D$ is not given, then this problem becomes combinatorial. If the set $D$ is not given but to be determined from a candidate set of states. Then a naive approach is to enumerate all possible combinations and evaluate the defender's value for every subset and select the one that yields the highest defender's value. It would be interesting to examine if the combinatorial problem is sub-modular or super-modular, but it is beyond the scope of this work.
\end{remark}

In summary, the proposed algorithm starts with an initial policy $\tilde \pi^{0}$, and use the \ac{irl}  to find the projection $\pi^0$ as well as their corresponding vectors   $(\vec{x}^{0}, \vec{y}^{0})$ that shape the attacker's perceptual reward function for which $\pi^0$ is optimal. Then a policy improvement is performed to update $\pi^0$ to $\tilde \pi^1$. By alternating the projection and policy improvement, the process terminates until the stopping criteria $|V_1^{k+1}(\nu) - V_1^{k}(\nu)| \le \epsilon$ is satisfied.

\section{Experiment}

We  illustrate the proposed methods with two sets of examples, one is a  probabilistic attack graph and another is an attack planning problem formulated in a stochastic gridworld. For all case studies, the workstation used is powered by Intel i7-11700K and 32GB RAM.

\begin{figure}[h]
	\centering
	\resizebox{0.6\linewidth}{!}{
		\begin{tikzpicture}[->,>=stealth',shorten >=1pt,auto,node distance=2  cm, scale =0.5,transform shape, line width = 0.5pt]
		\node[state, initial] (0) {$0$};
		\node[state] (2) [above right of=0] {$2$};
		\node[state] (3) [below right of=0] {$3$};
		\node[state] (1) [above of=2] {$1$};
		\node[state] (4) [below of=3] {$4$};
		\node[state] (5) [right of=1] {$5$};
		\node[state] (6) [right of=4] {$6$};
		\node[state] (7) [right of=3] {$7$};
		\node[state] (8) [right of=2] {$8$};
		\node[state] (9) [above right of=7] {$9$};
		\node[state] (10) [below right of=7] {$10$};
		\node[state] (13) [above right of=9] {$13$};
		\node[state] (12) [below right of=13] {$12$};
		\node[state] (11) [below right of=9] {$11$};
		
		\path (0) edge[line width = 1.2pt, color = black] node {} (1)
		(0) edge[] node {} (2)
		(0) edge[] node {} (3)
		(0) edge[]  node {} (4)
		(1) edge[line width = 1.2pt, color = black] node {} (5)
		(1) edge[] node {} (8)
		(1) edge[] node {} (6)
		(2) edge[line width = 1.2pt, color = black] node {} (6)
		(2) edge[] node {} (7)
		(3) edge[] node {} (5)
		(3) edge[] node {} (7)
		(4) edge[]  node {} (5)
		(4) edge[] node {} (7)
		(6) edge[] node {} (9)
		(6) edge[] node {} (10)
		(7) edge[line width = 1.2pt, color = black] node {} (9)
		(7) edge[] node {} (8)
		(9) edge[] node {} (11)
		(9) edge[] node {} (13)
		(10) edge[] node {} (11)
		(10) edge[bend right] node {} (12)
		(11) edge[bend left, line width = 1.2pt, color = black] node {} (12)
		(11) edge[] node {} (13)
		(12) edge[bend left] node {} (11)
		(12) edge[bend right] node {} (13)
		(13) edge[bend right, line width = 1.2pt, color = black] node {} (12)
		(13) edge[bend right] node {} (11)
		;
		\end{tikzpicture}
	}  
	\caption{A probabilistic attack graph.}
	\label{fig: simple probabilistic attack graph}
\end{figure}
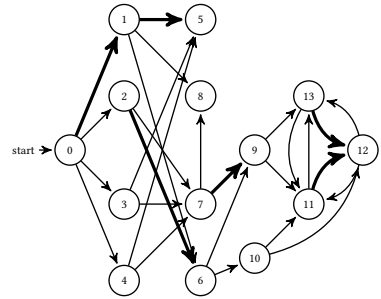

Figure~\ref{fig: simple probabilistic attack graph} shows a probabilistic attack graph with the target set   $F=\{10\}$ and the action set   $\{a,b,c,d\}$.  For clarity, the graph only shows the transition given action $a$ where a thick (resp. thin) arrow represents a high (resp. low) transition probability. For example, $P(0, a)=\{1:0.7, 2:0.1, 3:0.1, 4:0.1\}$ \footnote{The exact transition function is provided in the supplementary file.}.

Consider the set D = $\{11, 13\}$ of decoy states. Recall the defender's reward function is $R_1(s) = 1,$ for all $s \in D$. Assuming no resource is allocated to $D$ and all states in $D$ are sink states, then the attacker has a $60.33\%$ probability of reaching the target set $F$ from the initial state $0$. In the meantime, the defender's expected value is $0.149$. That is, with  probability $14.9\%$, the attacker will reach a decoy state in $D$ and the attack is terminated.

Given limited resource $\mathbf{1}^{\mathsf{T}}\vec{y} \le 3$, the decoy resource allocation yields $\vec{y}(11)=\vec{y}(13)=1.313$. Based on the given decoy resource allocation, the attacker has an $8.63\%$ probability of reaching the target set $F$ and the defender's expected reward is $0.653$ at initial state $0$. Thus, by assigning resources to decoys to attract the attacker, the defender reduces the attacker's probability of reaching the target state significantly ($85\%$ reduction) and improves the defender's value by 3.38 times.


\begin{figure}[ht]
    \centering
    \includegraphics[width=.5\linewidth]{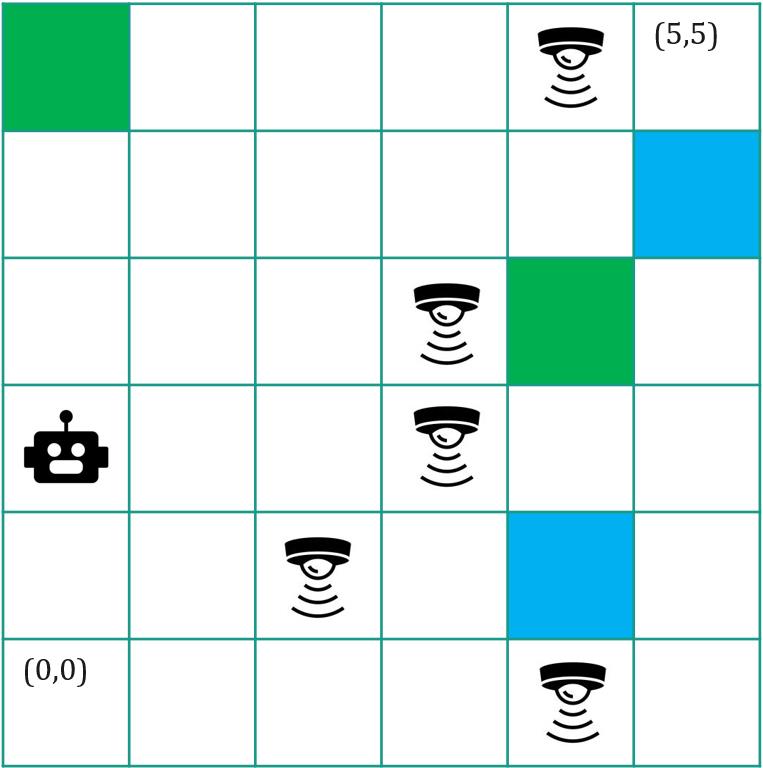}
    \caption{$6 \times 6$ gridworld example.}
    \label{fig:small gridworld}
\end{figure}

\begin{figure}[ht]
\begin{subfigure}{.45\textwidth}
    \centering
    \includegraphics[width=.8\linewidth]{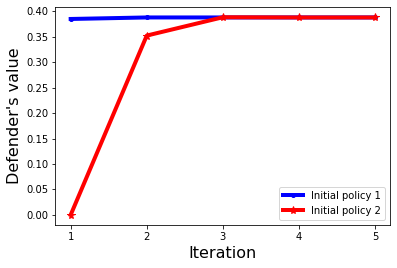}
    \caption{Converge  using different initial policies.}
    \label{fig:6*6 without action}
\end{subfigure}
\begin{subfigure}{.45\textwidth}
    \centering
    \includegraphics[width=.8\linewidth]{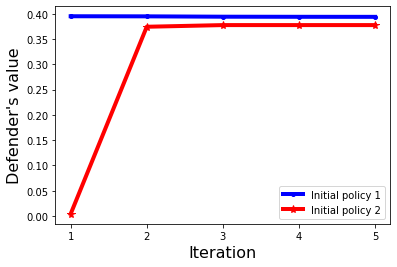}
    \caption{Converge   with decoy resource allocation and action reward modification.}
    \label{fig:6*6 with action}
\end{subfigure}
\label{fig: 6*6 result}
\caption{Defender's value converge trend in $6 \times 6$ gridworld example given $D = \{(1, 4), (4, 5)\}$.}
\end{figure}

\begin{figure}[ht]
    \centering
    \includegraphics[width=.8\linewidth]{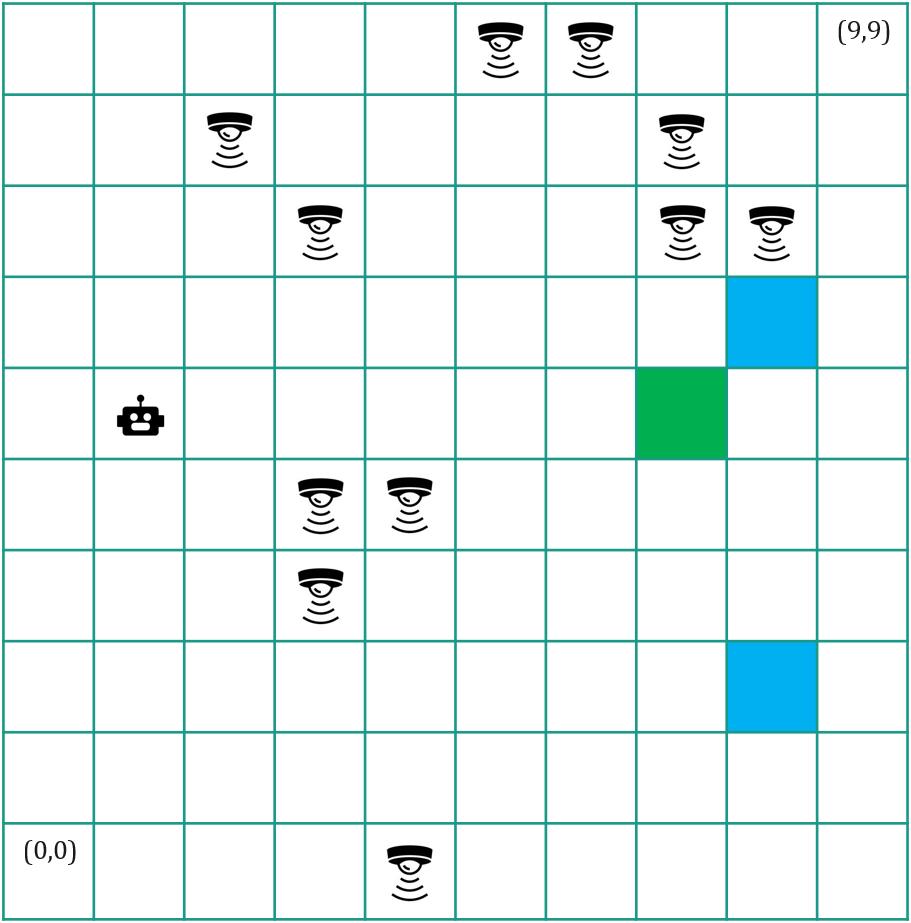}
    \caption{$10 \times 10$ gridworld example.}
    \label{fig:large gridworld}
\end{figure}

\begin{figure}[htbp]
    \centering
    \includegraphics[width=.8\linewidth]{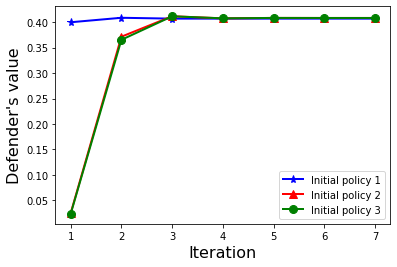}
    \caption{Defender's value converge trend in $10 \times 10$ gridworld. }
    \label{fig:def value in 10*10}
\end{figure}

Next, 
 we consider a robot motion planning problem in a stochastic $6 \times 6$ gridworld  shown in   Figure~\ref{fig:small gridworld}. The attacker/robot aims to reach a set of goal states while avoiding detection from the defender.  
 The attacker can move in four compass directions. Given an action, say, ``N'', the attacker enters the intended cell with $1 - 2\alpha$ probability, and enters the neighboring cells, which are  west and east cells with $\alpha$ probability. In our experiments, $\alpha$ is selected to be $0.1$. A state $(i,j)$ means the cell at row $i$ and column $j$.
 
The defender has deployed sensors shown in Figure~\ref{fig:small gridworld} to detect the presence of an attacker. Thus, once the attacker enters a sensor state, his task fails. The decoy set $D$ is given as blue cells and the target set $F$ is given as green cells.   

Given the initial state is at $(2, 0)$, which is indicated by the robot in the figure. We test the following three scenarios: No decoy resource allocation, decoy resource allocation only, decoy resource allocation together with reward modification. The result is shown in table~\ref{tab:6*6 result 1}. When we do not allocate resources to decoys, the attacker has a $98.98\%$ probability of reaching the target set $F$ while avoiding sensor states. And the defender's expected value is $3.56 \times 10^{-6}$. When the defender is allowed to allocate resources with a total budget of $4$ to decoys, the decoy resource allocation yields $\vec{y}((1, 4)) = 2.016, \vec{y}((4, 5)) = 1.826$.  The defender does not spend all decoy resources because of the use of  the logarithmic barrier function to enforce the constraint, when it is close to the upper bound, the log barrier function will work as a large penalty in gradient ascent.
 
Under the given resource allocation, the attacker has a $9.9\%$ probability of reaching the target set $F$, and the defender's expected value at the initial state is $0.3877$. In the decoy resource allocation and action reward modification experiment, the defender is allowed to modify all action rewards at state $(4, 4)$ and the action `N' reward at state $(4, 0), (4, 1)$ and $(4, 2)$. It turns out the defender allocates $1.938$ to decoy $(1, 4)$ and $1.734$ to decoy $(2, 5)$. Meanwhile, action `N'  reward at $(4, 0)$ is modified to $-1$ and the same action at $(4, 1)$ is modified to $-0.94$ and the action "N" reward at $(4, 2)$ is modified to $-0.904$, the defender will also modify the action reward of "W", "S", "N" at $(4, 4)$ to $-1$. Compare the decoy resource allocation result with the decoy resource allocation and action reward modification result. We find that by allowing action reward modification, the defender reduces the attacker's probability of reaching the target ($13.13\%$ reduction). In the meantime, the defender's expected value increases by $1.62\%$.

It is noted that due to the nonlinearity in the optimization problem,  the algorithm converges to different solutions under different initial conditions, as shown in Figures ~\ref{fig:6*6 without action} and ~\ref{fig:6*6 without action}. In the figures, the initial policy $1$ is generated by assuming the attacker receives the reward of $1$ if he reaches the decoy and receives a reward of $0$ when he reaches the target state. This is ideal for the defender's objective but is infeasible for the optimization problem  because in the attacker's perceptual planning problem, reaching the true target will always provide a reward of $1$ regardless of how many resources are allocated to decoys. 
The initial policy $2$ is randomly generated. In this experiment, the value of the objective function given different initial policies is close. 

In order to test how the decoy set $D$ influences the result. We re-allocate the position of decoys to $\{(0, 2), (5, 3)\}$. The result is shown in Table~\ref{tab:6*6 result 2}. Based on the new configuration, if we do not allocate decoy resources, the attacker reaches the target set with $98.97\%$ probability and the defender's value is $7.61 \times 10^{-8}$ at the initial state. If the defender can allocate resources to the decoys, our method yields $\vec{y}((0, 2)) = 1.141$ and $\vec{y}((5, 3)) = 1.0$. The attacker's probability of reaching the target set is $3.99\%$ and the defender's expected value is $0.6991$. If the defender is allowed to modify the same set of state-action rewards as she is in the previous example, our algorithm yields $\vec{y}((0, 2)) = 0.985$ and $\vec{y}((5, 3)) = 1.068$. Action `N'  reward at $(4, 0)$ is modified to $-1$ and the same action at $(4, 1)$ is modified to $-0.85$ and the action "N" reward at $(4, 2)$ is modified to $-0.081$, the defender will also modify the all action reward at $(4, 4)$ to $-1$. Under this configuration, the attacker's probability of reaching the target set is $0.286\%$ ($93\%$ reduction compared to only allocating decoy resources) and the defender's expected value is $0.7301$ ($4.4\%$ increase compared to only allocate decoy resources). By changing the configuration of set $D$, we show that the configuration of set $D$ influences the attacker's probability of reaching the target set and the defender's expected value: the second  set $D = \{(0,2), (5,3)\}$ appears to outperform the first set $D=\{(1,4), (4,5)\}$. 

\begin{table}[h]
\centering
\caption{Experiment result in $6 \times 6$ gridworld given $D = \{(1, 4), (4, 5)\}$.}
\resizebox{\columnwidth}{!}{%
\begin{tabular}{|c|c|c|c|}
\hline
\multicolumn{1}{|l|}{}                                                                     & \multicolumn{1}{l|}{No decoy  } & \multicolumn{1}{l|}{Decoy only} & \begin{tabular}[c]{@{}c@{}}Decoy and action reward \end{tabular} \\ \hline
\begin{tabular}[c]{@{}c@{}}Attacker's value\end{tabular} & 98.98\%                                         & 9.9\%                                        & 8.6\%                                                                                      \\ \hline
\begin{tabular}[c]{@{}c@{}}Defender's  value \end{tabular}      & 3.56 $\times$ $10^{-6}$                     & 0.3877                                       & 0.394                                                                                      \\ \hline
\end{tabular}
}

\label{tab:6*6 result 1}
\end{table}

\begin{table}[h]
\centering
\caption{Experiment result in $6 \times 6$ gridworld given $D = \{(0, 2), (5, 3)\}$.}
\resizebox{\columnwidth}{!}{%
\begin{tabular}{|c|c|c|c|}
\hline
\multicolumn{1}{|l|}{}                                                                     & \multicolumn{1}{l|}{No decoy} & \multicolumn{1}{l|}{Decoy only} & \begin{tabular}[c]{@{}c@{}}Decoy and action reward\end{tabular} \\ \hline
\begin{tabular}[c]{@{}c@{}}Attacker's value\end{tabular} & 98.97\%                                         & 3.99\%                                        & 0.286\%                                                                                      \\ \hline
\begin{tabular}[c]{@{}c@{}}Defender's  value \end{tabular}      & 7.61 $\times 10^{-8}$                     & 0.6991                                       & 0.7301                                                                                      \\ \hline
\end{tabular}
}
\label{tab:6*6 result 2}
\end{table}


Next, in order to test the scalability, we increase the gridworld size to $10 \times 10$ as shown in Figure~\ref{fig:large gridworld}. In the large gridworld example, we only do decoy resource allocation. The sensors, decoy set, and target set are represented using the same notation as the $6 \times 6$ gridworld. The defender's reward function is still $R_1(s) = 1$, for all $s \in D$. Assume the initial state is at $(5, 1)$. When the defender does not allocate decoy resources, the attacker's probability of reaching the target is $82.43\%$ and the defender's expected value at the initial state is $0.0024$. When the defender is allowed to allocate resources to decoys, our algorithm yields $\vec{y}((2, 8)) = 1.350, \vec{y}((6, 8)) = 1.235$. Under the given decoy resources, the attacker's probability of reaching the target decreases to $5.1\%$ ($94\%$ reduction), and the defender's expected value at the initial state increases to $0.4034$. We also test the defender's converging trend using different initial policies as shown in Figure~\ref{fig:def value in 10*10}. Initial policy $1$ is obtained similarly to initial policy $1$ in the $6 \times 6$ example. Initial policy $2$ and $3$ are randomly generated policies. From Figure~\ref{fig:def value in 10*10}, we observe that   different initial policies result in a similar converged value for the objective function. Considering the scalability of our algorithm, the computation time for the $10 \times 10$ gridworld example is $185.94$ seconds, while the computation time of the $6 \times     6$ example is $22.51$ seconds. The running time shows our algorithm can be extended to moderate problem sizes. It is noted that not only the state space size influences the running time but also the selection of decoys, the number of decoys influences the running time.


\section{Conclusion and Future Work}

We present a mathematical framework and algorithm for decoy allocation and reward modification in a proactive defense system. Our technical approach can be applied to many safety-critical systems where the probabilistic attack graphs are constructed from known vulnerabilities in a system. The formulation and solutions can be extended to a broad set of adversarial   interactions in which proactive defense with deception can be deployed.
In the future, we will consider more complex attack and defense objectives and investigate the decoy allocation given the uncertainty in the attacker's goal or capability. Apart from ``revealing the fake'' studied herein, we will also investigate how to ``conceal the truth'' by manipulating the attacker's perceptual reward of compromising true targets.

\bibliographystyle{ACM-Reference-Format.bst} 
\bibliography{refs.bib}


\begin{thebibliography}{21}


\ifx \showCODEN    \undefined \def \showCODEN     #1{\unskip}     \fi
\ifx \showDOI      \undefined \def \showDOI       #1{#1}\fi
\ifx \showISBNx    \undefined \def \showISBNx     #1{\unskip}     \fi
\ifx \showISBNxiii \undefined \def \showISBNxiii  #1{\unskip}     \fi
\ifx \showISSN     \undefined \def \showISSN      #1{\unskip}     \fi
\ifx \showLCCN     \undefined \def \showLCCN      #1{\unskip}     \fi
\ifx \shownote     \undefined \def \shownote      #1{#1}          \fi
\ifx \showarticletitle \undefined \def \showarticletitle #1{#1}   \fi
\ifx \showURL      \undefined \def \showURL       {\relax}        \fi
\providecommand\bibfield[2]{#2}
\providecommand\bibinfo[2]{#2}
\providecommand\natexlab[1]{#1}
\providecommand\showeprint[2][]{arXiv:#2}

\bibitem[\protect\citeauthoryear{{Anwar}, {Kamhoua}, and {Leslie}}{{Anwar}
  et~al\mbox{.}}{2020}]%
        {Anwar2020}
\bibfield{author}{\bibinfo{person}{A.~H. {Anwar}}, \bibinfo{person}{C.
  {Kamhoua}}, {and} \bibinfo{person}{N. {Leslie}}.}
  \bibinfo{year}{2020}\natexlab{}.
\newblock \showarticletitle{Honeypot Allocation over Attack Graphs in Cyber
  Deception Games}. In \bibinfo{booktitle}{\emph{2020 International Conference
  on Computing, Networking and Communications (ICNC)}}.
  \bibinfo{pages}{502--506}.
\newblock


\bibitem[\protect\citeauthoryear{Dempe and Zemkoho}{Dempe and Zemkoho}{2020}]%
        {dempe2020bilevel}
\bibfield{author}{\bibinfo{person}{Stephan Dempe} {and} \bibinfo{person}{Alain
  Zemkoho}.} \bibinfo{year}{2020}\natexlab{}.
\newblock \bibinfo{booktitle}{\emph{Bilevel optimization}}.
\newblock \bibinfo{publisher}{Springer}.
\newblock


\bibitem[\protect\citeauthoryear{Durkota, Lis{\`y}, Bo{\v{s}}ansk{\`y}, and
  Kiekintveld}{Durkota et~al\mbox{.}}{2015}]%
        {durkota2015optimal}
\bibfield{author}{\bibinfo{person}{Karel Durkota}, \bibinfo{person}{Viliam
  Lis{\`y}}, \bibinfo{person}{Branislav Bo{\v{s}}ansk{\`y}}, {and}
  \bibinfo{person}{Christopher Kiekintveld}.} \bibinfo{year}{2015}\natexlab{}.
\newblock \showarticletitle{Optimal network security hardening using attack
  graph games}. In \bibinfo{booktitle}{\emph{Twenty-Fourth International Joint
  Conference on Artificial Intelligence}}.
\newblock


\bibitem[\protect\citeauthoryear{Hansen, Jaumard, and Savard}{Hansen
  et~al\mbox{.}}{1992}]%
        {hansen1992new}
\bibfield{author}{\bibinfo{person}{Pierre Hansen}, \bibinfo{person}{Brigitte
  Jaumard}, {and} \bibinfo{person}{Gilles Savard}.}
  \bibinfo{year}{1992}\natexlab{}.
\newblock \showarticletitle{New branch-and-bound rules for linear bilevel
  programming}.
\newblock \bibinfo{journal}{\emph{SIAM Journal on scientific and Statistical
  Computing}} \bibinfo{volume}{13}, \bibinfo{number}{5} (\bibinfo{year}{1992}),
  \bibinfo{pages}{1194--1217}.
\newblock


\bibitem[\protect\citeauthoryear{Hong and Kim}{Hong and Kim}{2012}]%
        {hongHARMsHierarchicalAttack2012}
\bibfield{author}{\bibinfo{person}{Jin Hong} {and} \bibinfo{person}{Dong-Seong
  Kim}.} \bibinfo{year}{2012}\natexlab{}.
\newblock \showarticletitle{{{HARMs}}: {{Hierarchical Attack Representation
  Models}} for {{Network Security Analysis}}}. In
  \bibinfo{booktitle}{\emph{Australian {{Information Security Management
  Conference}}}}. \bibinfo{publisher}{{SRI Security Research Institute, Edith
  Cowan University, Perth, Western Australia}}, \bibinfo{pages}{9}.
\newblock


\bibitem[\protect\citeauthoryear{Hong and Kim}{Hong and Kim}{2016}]%
        {hongAssessingEffectivenessMoving2016}
\bibfield{author}{\bibinfo{person}{Jin~B. Hong} {and}
  \bibinfo{person}{Dong~Seong Kim}.} \bibinfo{year}{2016}\natexlab{}.
\newblock \showarticletitle{Assessing the {{Effectiveness}} of {{Moving Target
  Defenses Using Security Models}}}.
\newblock \bibinfo{journal}{\emph{IEEE Transactions on Dependable and Secure
  Computing}} \bibinfo{volume}{13}, \bibinfo{number}{2} (\bibinfo{date}{March}
  \bibinfo{year}{2016}), \bibinfo{pages}{163--177}.
\newblock


\bibitem[\protect\citeauthoryear{Jha, Sheyner, and Wing}{Jha
  et~al\mbox{.}}{2002}]%
        {jhaTwoFormalAnalyses2002}
\bibfield{author}{\bibinfo{person}{S. Jha}, \bibinfo{person}{O. Sheyner}, {and}
  \bibinfo{person}{J. Wing}.} \bibinfo{year}{2002}\natexlab{}.
\newblock \showarticletitle{Two Formal Analyses of Attack Graphs}. In
  \bibinfo{booktitle}{\emph{Proceedings 15th {{IEEE Computer Security
  Foundations Workshop}}. {{CSFW}}-15}}. \bibinfo{pages}{49--63}.
\newblock


\bibitem[\protect\citeauthoryear{Lallie, Debattista, and Bal}{Lallie
  et~al\mbox{.}}{2020}]%
        {lallieReviewAttackGraph2020a}
\bibfield{author}{\bibinfo{person}{Harjinder~Singh Lallie},
  \bibinfo{person}{Kurt Debattista}, {and} \bibinfo{person}{Jay Bal}.}
  \bibinfo{year}{2020}\natexlab{}.
\newblock \showarticletitle{A review of attack graph and attack tree visual
  syntax in cyber security}.
\newblock \bibinfo{journal}{\emph{Computer Science Review}}
  \bibinfo{volume}{35} (\bibinfo{date}{Feb.} \bibinfo{year}{2020}),
  \bibinfo{pages}{100219}.
\newblock
\showISSN{1574-0137}
\urldef\tempurl%
\url{https://doi.org/10.1016/j.cosrev.2019.100219}
\showDOI{\tempurl}


\bibitem[\protect\citeauthoryear{Madani}{Madani}{2002}]%
        {Omid2022aaai}
\bibfield{author}{\bibinfo{person}{Omid Madani}.}
  \bibinfo{year}{2002}\natexlab{}.
\newblock \showarticletitle{On Policy Iteration as a Newton's Method and
  Polynomial Policy Iteration Algorithms}. In
  \bibinfo{booktitle}{\emph{Eighteenth National Conference on Artificial
  Intelligence}} (Edmonton, Alberta, Canada). \bibinfo{publisher}{American
  Association for Artificial Intelligence}, \bibinfo{address}{USA},
  \bibinfo{pages}{273–278}.
\newblock
\showISBNx{0262511290}


\bibitem[\protect\citeauthoryear{Milani, Shen, Chan, Venkatesan, Leslie,
  Kamhoua, and Fang}{Milani et~al\mbox{.}}{2020a}]%
        {milaniHarnessingPowerDeception2020}
\bibfield{author}{\bibinfo{person}{Stephanie Milani}, \bibinfo{person}{Weiran
  Shen}, \bibinfo{person}{Kevin~S. Chan}, \bibinfo{person}{Sridhar Venkatesan},
  \bibinfo{person}{Nandi~O. Leslie}, \bibinfo{person}{Charles Kamhoua}, {and}
  \bibinfo{person}{Fei Fang}.} \bibinfo{year}{2020}\natexlab{a}.
\newblock \showarticletitle{Harnessing the {Power} of {Deception} in {Attack}
  {Graph}-{Based} {Security} {Games}}. In \bibinfo{booktitle}{\emph{Decision
  and {Game} {Theory} for {Security}}} \emph{(\bibinfo{series}{Lecture {Notes}
  in {Computer} {Science}})}, \bibfield{editor}{\bibinfo{person}{Quanyan Zhu},
  \bibinfo{person}{John~S. Baras}, \bibinfo{person}{Radha Poovendran}, {and}
  \bibinfo{person}{Juntao Chen}} (Eds.). \bibinfo{publisher}{Springer
  International Publishing}, \bibinfo{address}{Cham},
  \bibinfo{pages}{147--167}.
\newblock
\showISBNx{978-3-030-64793-3}
\urldef\tempurl%
\url{https://doi.org/10.1007/978-3-030-64793-3_8}
\showDOI{\tempurl}


\bibitem[\protect\citeauthoryear{Milani, Shen, Chan, Venkatesan, Leslie,
  Kamhoua, and Fang}{Milani et~al\mbox{.}}{2020b}]%
        {milani2020harnessing}
\bibfield{author}{\bibinfo{person}{Stephanie Milani}, \bibinfo{person}{Weiran
  Shen}, \bibinfo{person}{Kevin~S Chan}, \bibinfo{person}{Sridhar Venkatesan},
  \bibinfo{person}{Nandi~O Leslie}, \bibinfo{person}{Charles Kamhoua}, {and}
  \bibinfo{person}{Fei Fang}.} \bibinfo{year}{2020}\natexlab{b}.
\newblock \showarticletitle{Harnessing the power of deception in attack
  graph-based security games}. In \bibinfo{booktitle}{\emph{International
  Conference on Decision and Game Theory for Security}}. Springer,
  \bibinfo{pages}{147--167}.
\newblock


\bibitem[\protect\citeauthoryear{Nguyen, Wright, Wellman, and Singh}{Nguyen
  et~al\mbox{.}}{2018}]%
        {nguyenMultistageAttackGraph2018}
\bibfield{author}{\bibinfo{person}{Thanh~H. Nguyen}, \bibinfo{person}{Mason
  Wright}, \bibinfo{person}{Michael~P. Wellman}, {and}
  \bibinfo{person}{Satinder Singh}.} \bibinfo{year}{2018}\natexlab{}.
\newblock \showarticletitle{Multistage {{Attack Graph Security Games}}:
  {{Heuristic Strategies}}, with {{Empirical Game}}-{{Theoretic Analysis}}}.
\newblock \bibinfo{journal}{\emph{Security and Communication Networks}}
  \bibinfo{volume}{2018} (\bibinfo{date}{Dec.} \bibinfo{year}{2018}),
  \bibinfo{pages}{1--28}.
\newblock


\bibitem[\protect\citeauthoryear{Noel and Jajodia}{Noel and Jajodia}{2008}]%
        {noel2008optimal}
\bibfield{author}{\bibinfo{person}{Steven Noel} {and} \bibinfo{person}{Sushil
  Jajodia}.} \bibinfo{year}{2008}\natexlab{}.
\newblock \showarticletitle{Optimal ids sensor placement and alert
  prioritization using attack graphs}.
\newblock \bibinfo{journal}{\emph{Journal of Network and Systems Management}}
  \bibinfo{volume}{16}, \bibinfo{number}{3} (\bibinfo{year}{2008}),
  \bibinfo{pages}{259--275}.
\newblock


\bibitem[\protect\citeauthoryear{Noel, Jajodia, Wang, and Singhal}{Noel
  et~al\mbox{.}}{2010}]%
        {noel2010measuring}
\bibfield{author}{\bibinfo{person}{Steven Noel}, \bibinfo{person}{Sushil
  Jajodia}, \bibinfo{person}{Lingyu Wang}, {and} \bibinfo{person}{Anoop
  Singhal}.} \bibinfo{year}{2010}\natexlab{}.
\newblock \showarticletitle{Measuring security risk of networks using attack
  graphs}.
\newblock \bibinfo{journal}{\emph{International Journal of Next-Generation
  Computing}} \bibinfo{volume}{1}, \bibinfo{number}{1} (\bibinfo{year}{2010}),
  \bibinfo{pages}{135--147}.
\newblock


\bibitem[\protect\citeauthoryear{Puterman}{Puterman}{2014}]%
        {puterman2014markov}
\bibfield{author}{\bibinfo{person}{Martin~L Puterman}.}
  \bibinfo{year}{2014}\natexlab{}.
\newblock \bibinfo{booktitle}{\emph{Markov decision processes: discrete
  stochastic dynamic programming}}.
\newblock \bibinfo{publisher}{John Wiley \& Sons}.
\newblock


\bibitem[\protect\citeauthoryear{Sheyner, Haines, Jha, Lippmann, and
  Wing}{Sheyner et~al\mbox{.}}{2002}]%
        {sheyner2002automated}
\bibfield{author}{\bibinfo{person}{Oleg Sheyner}, \bibinfo{person}{Joshua
  Haines}, \bibinfo{person}{Somesh Jha}, \bibinfo{person}{Richard Lippmann},
  {and} \bibinfo{person}{Jeannette~M Wing}.} \bibinfo{year}{2002}\natexlab{}.
\newblock \showarticletitle{Automated generation and analysis of attack
  graphs}. In \bibinfo{booktitle}{\emph{Proceedings 2002 IEEE Symposium on
  Security and Privacy}}. IEEE, \bibinfo{pages}{273--284}.
\newblock


\bibitem[\protect\citeauthoryear{Singhal and Ou}{Singhal and Ou}{2017}]%
        {singhal2017security}
\bibfield{author}{\bibinfo{person}{Anoop Singhal} {and}
  \bibinfo{person}{Xinming Ou}.} \bibinfo{year}{2017}\natexlab{}.
\newblock \showarticletitle{Security risk analysis of enterprise networks using
  probabilistic attack graphs}.
\newblock In \bibinfo{booktitle}{\emph{Network Security Metrics}}.
  \bibinfo{publisher}{Springer}, \bibinfo{pages}{53--73}.
\newblock


\bibitem[\protect\citeauthoryear{SINHA, FANG, AN, KIEKINTVELD, and TAMBE}{SINHA
  et~al\mbox{.}}{2018}]%
        {sinhaStackelbergSecurityGames2018}
\bibfield{author}{\bibinfo{person}{Arunesh SINHA}, \bibinfo{person}{Fei FANG},
  \bibinfo{person}{Bo AN}, \bibinfo{person}{Christopher KIEKINTVELD}, {and}
  \bibinfo{person}{Milind TAMBE}.} \bibinfo{year}{2018}\natexlab{}.
\newblock \showarticletitle{Stackelberg Security Games: {{Looking}} beyond a
  Decade of Success}.
\newblock \bibinfo{journal}{\emph{Proceedings of the Twenty-Seventh
  International Joint Conference on Artificial Intelligence
  (IJCAI-18),Stockholm, Sweden, July 13-19}} (\bibinfo{date}{July}
  \bibinfo{year}{2018}), \bibinfo{pages}{5494--5501}.
\newblock


\bibitem[\protect\citeauthoryear{Thakoor, Tambe, Vayanos, Xu, Kiekintveld, and
  Fang}{Thakoor et~al\mbox{.}}{2019}]%
        {thakoorCyberCamouflageGames2019}
\bibfield{author}{\bibinfo{person}{Omkar Thakoor}, \bibinfo{person}{Milind
  Tambe}, \bibinfo{person}{Phebe Vayanos}, \bibinfo{person}{Haifeng Xu},
  \bibinfo{person}{Christopher Kiekintveld}, {and} \bibinfo{person}{Fei Fang}.}
  \bibinfo{year}{2019}\natexlab{}.
\newblock \showarticletitle{Cyber {{Camouflage Games}} for {{Strategic
  Deception}}}. In \bibinfo{booktitle}{\emph{Decision and {{Game Theory}} for
  {{Security}}}} \emph{(\bibinfo{series}{Lecture {{Notes}} in {{Computer
  Science}}})}, \bibfield{editor}{\bibinfo{person}{Tansu Alpcan},
  \bibinfo{person}{Yevgeniy Vorobeychik}, \bibinfo{person}{John~S. Baras},
  {and} \bibinfo{person}{Gy{\"o}rgy D{\'a}n}} (Eds.).
  \bibinfo{publisher}{{Springer International Publishing}},
  \bibinfo{address}{{Cham}}, \bibinfo{pages}{525--541}.
\newblock
\showISBNx{978-3-030-32430-8}
\urldef\tempurl%
\url{https://doi.org/10.1007/978-3-030-32430-8_31}
\showDOI{\tempurl}


\bibitem[\protect\citeauthoryear{Wu, Li, Liu, Bao, Zheng, Ye, and Luo}{Wu
  et~al\mbox{.}}{2019}]%
        {wu2019reward}
\bibfield{author}{\bibinfo{person}{Guojun Wu}, \bibinfo{person}{Yanhua Li},
  \bibinfo{person}{Zhenming Liu}, \bibinfo{person}{Jie Bao},
  \bibinfo{person}{Yu Zheng}, \bibinfo{person}{Jieping Ye}, {and}
  \bibinfo{person}{Jun Luo}.} \bibinfo{year}{2019}\natexlab{}.
\newblock \showarticletitle{Reward Advancement: Transforming Policy under
  Maximum Causal Entropy Principle}.
\newblock \bibinfo{journal}{\emph{arXiv preprint arXiv:1907.05390}}
  (\bibinfo{year}{2019}).
\newblock


\bibitem[\protect\citeauthoryear{Ziebart, Maas, Bagnell, Dey,
  et~al\mbox{.}}{Ziebart et~al\mbox{.}}{2008}]%
        {ziebart2008maximum}
\bibfield{author}{\bibinfo{person}{Brian~D Ziebart}, \bibinfo{person}{Andrew~L
  Maas}, \bibinfo{person}{J~Andrew Bagnell}, \bibinfo{person}{Anind~K Dey},
  {et~al\mbox{.}}} \bibinfo{year}{2008}\natexlab{}.
\newblock \showarticletitle{Maximum entropy inverse reinforcement learning.}.
  In \bibinfo{booktitle}{\emph{Aaai}}, Vol.~\bibinfo{volume}{8}. Chicago, IL,
  USA, \bibinfo{pages}{1433--1438}.
\newblock


\end{thebibliography}

\end{document}